\documentclass[conference]{IEEEtran}
\IEEEoverridecommandlockouts
\ifCLASSINFOpdf

\else

\fi
\usepackage{cite}
\usepackage{amsmath,amssymb,amsfonts,amsthm}

\usepackage{graphicx}
\usepackage{textcomp}
\usepackage{color}
\usepackage{epstopdf}
\usepackage{listings}
\usepackage{float}
\usepackage{balance}
\usepackage{color} 
\definecolor{mygreen}{RGB}{28,172,0} 
\definecolor{mylilas}{RGB}{170,55,241}
\usepackage{float} 
\newtheorem{theorem}{\textbf{Theorem}}

\usepackage{mathtools} 
\usepackage{cuted} 
\usepackage{hyperref}

\newtheorem{note}{\textbf{Note}}
\usepackage{soul}
\usepackage{algorithm,algpseudocode}
\algrenewcommand\algorithmicindent{0.7em}
\usepackage[labelsep=space, font=small, labelfont=bf]{caption} 
\usepackage{subcaption}
\usepackage{xcolor}
\DeclareMathOperator{\diag}{diag}

\def\BibTeX{{\rm B\kern-.05em{\sc i\kern-.025em b}\kern-.08em
    T\kern-.1667em\lower.7ex\hbox{E}\kern-.125emX}}
    
\begin{document}

\title{Gaussian Integral based Bayesian Smoother
}

\author{\IEEEauthorblockN{ Rohit Kumar Singh}
\IEEEauthorblockA{\textit{Department of Electrical Engineering} \\
\textit{Indian Institute of Technology Patna}\\
Bihar, India \\
rohit\_1921ee19@iitp.ac.in}
\and
\IEEEauthorblockN{ Kundan Kumar}
\IEEEauthorblockA{\textit{Department of Electrical Engineering}\\\textit{and Automation}\\
\textit{Aalto University, Finalnd}\\
kundan.kumar@aalto.fi}
\and
\IEEEauthorblockN{ Shovan Bhaumik}
\IEEEauthorblockA{\textit{Department of Electrical Engineering} \\
\textit{Indian Institute of Technology Patna}\\
Bihar, India \\
shovan.bhaumik@iitp.ac.in}

}

\maketitle

\begin{abstract}
This work introduces the Gaussian integration to address a smoothing problem of a nonlinear stochastic state space model. The probability densities of states at each time instant are assumed to be Gaussian, and their means and covariances are evaluated by utilizing the odd-even properties of Gaussian integral, which are further utilized to realize Rauch-Tung-Striebel (RTS) smoothing expressions. Given that the Gaussian integration provides an exact solution for the integral of a polynomial function over a Gaussian probability density function, it is anticipated to provide more accurate results than other existing Gaussian approximation-based smoothers such as extended Kalman, cubature Kalman, and unscented Kalman smoothers, especially when polynomial types of nonlinearity are present in the state space models. The developed smoothing algorithm is applied to the Van der Pol oscillator, where the nonlinearity associated with their dynamics is represented using polynomial functions. Simulation results are provided to demonstrate the superiority of the proposed algorithm.
\end{abstract}

\section{Introduction}
In this paper, we developed an approximated Gaussian integral-based smoothing algorithm for the nonlinear stochastic state space models. The Smoothing algorithms find application in various areas like target tracking \cite{xia2022multiple}, navigation \cite{cao2022gvins}, robotics \cite{garcia2019rao}, image analysis \cite{ritter2021data}, statistical modelling \cite{munezero2022efficient}, and finances \cite{zhang2023sequential}.
Here, we consider a discrete-time stochastic state space model in the following form \cite{sarkka2023bayesian}
	\begin{align} \label{eqnprocess}
			x_k =& f(x_{k-1})+\mu_{k-1}, \\ \label{eqnmeasurement}
			y_k =& h(x_k)+\nu_k,
	\end{align}
where $x_k\in  \mathbb{R}^n$  and $y_k \in \mathbb{R}^m$ are the state of the system and measurement, respectively, at $k^{th}$ instant, $k\in 1,2,\ldots, T$. 
The state evaluation function, $f(x): \mathbb{R}^n \rightarrow \mathbb{R}^n$ and measurement mapping $h(x): \mathbb{R}^n \rightarrow \mathbb{R}^m$ are known nonlinear functions. 
The process and measurement noises are uncorrelated, and they follow a Gaussian probability density function (pdf) given as $\mu_{k-1}\sim \mathcal{N}(0, Q_{k-1})$ and $\nu_k\sim\mathcal{N}(0, R_k)$, respectively, where $Q_{k-1}$ and $R_k$ are the process and measurement noise covariance matrix, respectively.

Smoothing in the context of estimation refers to Bayesian methodology for estimating the past state's history of a system using the sensor noisy measurements available up to the final time step. 
It enhances the estimation performance by refining earlier state estimates obtained from the filtering \cite{sarkka2023bayesian}. The smoother computes the marginal posterior distribution of the state ($x_k$) given measurement up to time step $T$  \emph{i.e.} $p(x_k|y_{1:T})$ in two steps: (i) forward filtering and (ii) backward filtering. The forward filter step involves recursively computation of the posterior state estimate $p(x_k|y_k)$ \cite{bar2004estimation, sarkka2023bayesian} 
\begin{align}
p(x_k|y_{k-1}) &= \int p(x_k|x_{k-1})p(x_{k-1}|y_{1:k-1})dx_{k-1} \label{eqnpredict1} \\
p(x_k|y_k) & \propto p(y_k|x_k)p(x_k|y_{1:k-1}) \label{eqnpredict2}.
\end{align}
The Backward filtering step computes $p(x_k|y_{1:T})$ recursively backwards starting from the final time step ($T$) 
 \begin{equation} \label{eqnsmoothint}
     p(x_k|y_{1:T})= p(x_k|y_{1:k})\int \dfrac{p(x_{k+1}|x_k)p(x_{k+1}|y_{1:T})}{p(x_{k+1}|y_{1:k})}dx_{k+1}.
 \end{equation}
For the linear systems, the distribution remains Gaussian, and the Rauch-Tung-Striebel smoother (RTSS) provide a closed-form solution \cite{rauch1965maximum, sarkka2023bayesian}. However, for nonlinear systems, the distributions no longer remain Gaussian, and the closed-form solutions are not available \cite{ito2000gaussian}. In literature, many times, using moment matching, these distributions are approximated to follow the Gaussian distribution \cite{sarkka2023bayesian,ito2000gaussian,wan2000unscented, sarkka2008unscented, arasaratnam2009cubature}, and subsequently mean and their covariances are being computed. In this work, we use these Gaussian approximation-based methods. Other approximation methods also exist, such as the sequential Monte Carlo \cite{bolviken2001monte} and other related approaches. 

Various Gaussian approximation-based smoothing algorithms exist in the literature \cite{sarkka2023bayesian}, such as Taylor series approximation-based extended RTS smoother (ERTSS) \cite{leondes1970nonlinear} and numerical approximation-based smoother, including the unscented RTSS (URTSS) \cite{sarkka2008unscented}, cubature RTSS (CRTSS) \cite{arasaratnam2011cubature}, Gauss-Hermite RTSS \cite{sarkka2010gaussian}, Fourier-Hermite RTSS \cite{sarmavuori2012fourier}, among others. 
In these Gaussian approximation-based RTS smoother, to realize them, the integrals in \eqref{eqnpredict1}-\eqref{eqnsmoothint} need to be evaluated. If we look closely, these integrals are of the form $\int nonlinear function \times Gaussian \; pdf$. In this work, we propose to solve the smoothing problems using the Gaussian integral as developed in \cite{kumar2023new}. The Gaussian integral provides an exact solution for the integration mentioned in RTS smoothing for a non-linear polynomial function. So, if we use Gaussian integral to develop a smoother for a nonlinear system having polynomial type of nonlinearity, it is expected to perform better in terms of estimation accuracy than existing Gaussian approximation smoothing algorithms \cite{sarkka2023bayesian}. 
When the nonlinear function is not of polynomial type, the developed method can still be applied after approximating the function using the Taylor series expansion. However, the integration with such expansion will no longer be exact as we will have to use the truncated Taylor series expansion. 

To demonstrate the effectiveness of the proposed smoothing algorithm, we apply it to the Van der Pol (VDP) oscillator and compare the results with existing methods. The VDP oscillator, known for its non-linear damping and complex dynamics, is widely used in fields such as electrical engineering, biology \cite{fitzhugh1961impulses}, and physics \cite{onat2019novel}.



\section{Backward filtering in Bayesian smoother}
Non-linear forward filtering is the process of estimating the hidden variable of a time-varying system using noisy sensor measurements, conditioned that the system is observable \cite{yamalakonda2023oscillatory}.
The integrations involved in calculating the posterior state estimate and covariance matrix during the forward filtering are provided in \cite[pp.2]{dunik2013stochastic} or \cite[pp.3]{arasaratnam2009cubature}. 

The backward filtering procedure for the RTS smoother in the Bayesian framework with Gaussian assumptions is derived below.
The joint distribution of $x_k$ and $x_{k+1}$, given the measurements available up to instant $k$ is given as
	\begin{equation}
		p(x_k,x_{k+1}|y_{1:k}) = p(x_{k+1}|x_{1:k})p(x_{k}|y_{1:k}),
	\end{equation}   
	where $p(x_{k+1}|x_{1:k})\sim \mathcal{N}(x_{k+1}; \hat{x}_{k+1|k},P_{xx,k+1|k})$ and $p(x_{k}|y_{1:k})\sim \mathcal{N}(x_{k};\hat{x}_{k|k},P_{xx,k|k})$ are pdf for predicted states and posterior state estimate, obtained from the forward filtering. Hence, the joint pdf of  $x_k$ and $x_{k+1}$ conditioned over measurement $y_{1:k}$ using \cite[lemma~A.2]{sarkka2023bayesian}, is given as
	\begin{equation} \label{eqnrts2}
		p(\begin{bmatrix}  x_k \\ x_{k+1}  \end{bmatrix}|y_{1:k}) \sim \mathcal{N}( \begin{bmatrix} \hat{x}_{k|k} \\ \hat{x}_{k+1|k} \end{bmatrix}, \begin{bmatrix} P_{xx,k|k} & P_{x_k,x_{k+1}} \\ P_{x_k,x_{k+1}}^\top & P_{xx,k+1|k}
		\end{bmatrix}   ),
	\end{equation} 
	where the cross-covariance matrix ($P_{x_k,x_{k+1}}$) is \cite[p.~279]{sarkka2023bayesian}
	\begin{equation}\begin{split} \label{eqnrtscross}
		&P_{x_k,x_{k+1}}=  \mathbb{E}[(x_k-\hat{x}_{k|k})(x_{k+1}-\hat{x}_{k+1|k})^\top  ],\\
		&=  \int x_k x_{k+1}^\top\mathcal{N}(x_k;\hat{x}_{k|k},P_{xx,k|k})dx_k - \hat{x}_{k|k}\hat{x}_{k+1|k}^\top,\\
		&=  \int x_k f(x_{k})^\top\mathcal{N}(x_k;\hat{x}_{k|k},P_{xx,k|k})dx_k -  \hat{x}_{k|k}\hat{x}_{k+1|k}^\top.
    \end{split}	\end{equation}
	The smoothing pdf for time step $(k+1)$ is $p(x_{k+1}|y_{1:T})$, which is assumed to be known and Gaussian, is given as
	\begin{equation} \label{eqnrts4}
		p(x_{k+1}|y_{1:T})\sim \mathcal{N}(x_{k+1};\hat{x}_{k+1|T}^s,P_{xx,k+1|T}^s).
	\end{equation}
Following \cite[pp.~255--257]{sarkka2023bayesian}, we obtain the smoothing distribution at time step $k$, $p(x_k|y_{1:T})\sim \mathcal{N}(x_k;\hat{x}_{k|T}^s,P_{xx,k|T}^s)$, where 
\begin{equation} \begin{split} \label{eqnrtsx}
		\hat{x}_{k|T}^s = & \hat{x}_{k|k}+P_{x_k,x_{k+1}}P_{xx,k+1|k}^{-1}(\hat{x}_{k+1|T}^s-\hat{x}_{k+1|k}),
		\end{split}\end{equation}
  \begin{equation} \begin{split} \label{eqnrtsp}
  P_{xx,k|T}^s
		 = & P_{xx,k|k}+P_{x_k,x_{k+1}}P_{xx,k+1|k}^{-1}(P_{xx,k+1|T}^s- \\
   & P_{xx,k+1|k})(P_{x_k.x_{k+1}}P_{xx,k+1|k}^{-1})^\top.
	\end{split}\end{equation}
From the above expressions, we see that the backward filtering is also a recursive process proceeding backwards in time starting from the final time step $T$, such that the smoothing distribution of states at time $k$ is derived from the smoothing distribution at $(k+1)$. 
The parameters such as $\hat{x}_{k|k}$, $P_{xx,k|k}$, $\hat{x}_{k+1|k}$, and $P_{xx,k+1|k}$,  in the \eqref{eqnrtsx} and \eqref{eqnrtsp} are fetched from the forward filtering. $\hat{x}_{k+1|T}^s$ and $P_{xx,k+1|T}^s$ is known and initialized with the posterior estimate at the final time step $T$, \emph{i.e} $\hat{x}_{T|T}^s=\hat{x}_{T|T},P_{xx,T|T}^s=P_{xx,T|T}$. $P_{x_k,x_{k+1}}$ is calculated using \eqref{eqnrtscross}.

\section{Smoothing using Gaussian Integral}
In this section, we propose a Gaussian-integral RTS smoother (GIRTSS) for the stochastic state-space model from \eqref{eqnprocess} and \eqref{eqnmeasurement}.
The realisation of the smoothing algorithm demands the evaluation of integrals from forward filtering mentioned in \cite[pp.3]{arasaratnam2009cubature} and \eqref{eqnrtscross} from backward filtering. For that, as discussed earlier, we need to evaluate the integrals of $f(x_k)$ and $x_kf(x_k)$ over a Gaussian pdf. 
This is done using the Gaussian integration (GI) for simple polynomial functions.

The integration of the function $y=\exp(-\dfrac{1}{2}a x^2)$, where the plot for $x$ \emph{vs.} $y$ will be bell-shaped, and the width of curve is influenced by value of parameter $a$, is \cite{straub2009brief} $\int\limits_{-\infty}^{\infty} \exp(-\dfrac{1}{2}ax^2) dx = \sqrt{\dfrac{2\pi}{a}}.$
	For any variable $x\in \mathbb{R}$, such that $x\sim \mathcal{N}(0,c)$ and $m$ is any non-negative integer, the integral is
	\begin{equation*} 
			\int\limits_{-\infty}^{\infty}x^m \exp(-\dfrac{x^2}{2c^2})dx=    \begin{cases} (2c)^{\dfrac{m+1}{2}}\Gamma(\dfrac{m+1}{2}) \ \  \text{if m is even}, \\ 
                0, \text{if m is odd}. \end{cases}
	\end{equation*}
	Similarly, for a variable $x\in \mathbb{R}^{n}$, where independent identically distributed $x=\begin{bmatrix} x_1 & x_2 \ldots x_n \end{bmatrix}^\top \sim \mathcal{N}(0,c_i)$, their corresponding orders are $m_i (i=1,2\ldots n)$, and variance are  $c_i (i=1,2\ldots n)$, the integral for such case is 
\begin{equation*} \begin{split}
			 &  \int\limits_{-\infty}^{\infty} x_1^{m_1} \exp(- \dfrac{x_1^2}{2c_1})x_2^{m_2} \exp(- \dfrac{x_2^2}{2c_2}) \ldots x_n^{m_n} \exp(- \dfrac{x_n^2}{2c_n}) dx,\\
			&=  \int\limits_{-\infty}^{\infty} \prod_{i=1}^{i=n} x_i^{m_i} \exp(-\dfrac{1}{2}\sum_{i=1}^{n}\dfrac{x_i^2}{c_i})dx,\\
     &=  \begin{cases}\prod_{i=1}^{n}\{ (2c_i)^{\dfrac{m_i+1}{2}}\Gamma(\dfrac{m_i+1}{2}) \} \ \ \text{if all $m_i$'s are even  },\\
			   0 \  \text{otherwise}.\end{cases}
	\end{split}	\end{equation*} 
 \begin{theorem}
	The integral of any polynomial function $f_1(x) = \prod_{i=1}^{n} x_i^{m_i}$ having Gaussian pdf $\mathcal{N}(x;\hat{x},P)$, where $x_i = \begin{bmatrix} x_1 & x_2 & \ldots & x_n \end{bmatrix}$ and $m_i\in Z+$ is
\begin{equation}\begin{split} \label{eqngi}
& \int\limits_{-\infty}^{\infty} f_1(x)\mathcal{N}(x;\hat{x},P)dx \\&
			= \begin{cases}  \dfrac{1}{\sqrt{{\pi}^n}} \Big[ \sum_{\cup_{a_{i,j}}} \Bigl\{ \prod_{i=1}^{n} \big( D_i \hat{x}_{i}^{a_{i,1}} ( \prod_{j=1}^{n}S_{ij}^{a_{i,j+1}} ) \\
   (2c_i)^{(\dfrac{1}{2}\sum_{j=1}^n a_{j,i+1} )} 
			 \Gamma(\dfrac{\sum_{j=1}^n a_{j,i+1}+1}{2})     \big)\Bigr\}\Big]	\\ \text{when}\ \sum_{j=1}^na_{j,i+1} \ \text{is even},\\
			   0 \ \  \text{otherwise}, \end{cases}
\end{split}\end{equation}
where $D_i=\dfrac{m_i!}{a_{i,1}!a_{i,2}\,!\ldots a_{i,n+1}!}$ is the multinomial coefficient. $c_i$ is the eigen value of the covariance matrix $P$, and $S$ is the orthogonal matrix satisfying the condition $S^{-1}PS=\diag(d_1,d_2\ldots d_n)$, $S_{i,j}$ is the $(i,j)^{th}$ element of matrix $S$. 
  The set $\cup_{a_{i,j}}$ $(i=1,2,\ldots n),\ (j=1,2\ldots n+1)$ are all possible combinations of $a_{i,j}$ which satisfies the two conditions, that is $\sum_{j=1}^{n+1}a_{i,j} = m_i$, and $\sum_{j=1}^na_{j,i+1} =\text{even}$, subsequently.
 \end{theorem}
 \begin{proof}
     The detailed proof for Theorem 1 is provided in \cite{kumar2023new}.
 \end{proof}
\vspace{-0.8 cm}
\begin{note}
Theorem 1 can be extended for the calculation of integration of multidimensional polynomial function of form $f(x)= \begin{bmatrix} f_1(x) & f_2(x)\ldots f_n(x) \end{bmatrix}^\top$, which are evaluated by calculating the integral row wise separately, \emph{i.e.} $\mathcal{I} = \begin{bmatrix} \mathcal{I}_1 & \mathcal{I}_2 & \ldots & \mathcal{I}_3 \end{bmatrix}$.
\end{note}
The backward filtering demands the evaluation of the integration of form $x_kf(x_k)$ for the calculation of the cross-covariance matrix as given in \eqref{eqnrtscross}.
The integrand from the cross-covariance \eqref{eqnrtscross} \emph{i.e.} $\int x_k f(x_k)^\top \mathcal{N}(x_k;\hat{x}_{k|k},\\P_{xx,k|k})dx_k$ can be written in expanded form for process function $f(x_k)=\begin{bmatrix} f_1(x_{k}) & f_2(x_{k}) \ldots f_n(x_{k}) \end{bmatrix}_{n \times 1}^\top$ as
\begin{equation} \label{eqnmatrix}
  x_kf(x_k)^\top = \begin{bmatrix} 
      x_{1,k} f_1(x_k) & x_{1,k} f_2(x_k) &  \ldots & x_{1,k} f_n(x_k) \\
      x_{2,k} f_1(x_k) & x_{2,k} f_2(x_k) &  \ldots & x_{2,k} f_n(x_k) \\
      \vdots & \vdots & \vdots & \vdots \\
      x_{n,k} f_1(x_k) & x_{n,k} f_2(x_k) &  \ldots & x_{n,k} f_n(x_k) \\
  \end{bmatrix},
\end{equation}
where the process function vector elements are non-linear polynomial functions which are written as a sum of polynomial terms as, $f_{g'}(x)= \sum_{p=1}^{q} \alpha_p \prod_{i=1}^n x_{i}^{m_{i,p}} $.
\begin{theorem}
    The Gaussian integration for any element of matrix \eqref{eqnmatrix}  can be computed by 
    \begin{equation} \label{eqnfg}
        \mathcal{I} =  \int\limits_{-\infty}^{\infty} x_{g,k}f_{g'}(x_k)\mathcal{N}(x_k;\hat{x}_{k|k},P_{xx,k|k})dx_k,
    \end{equation}
    \vspace{-0.5 cm}
\begin{equation} \label{eqngi2}
        \begin{split}
           &=  \sum_{p=1}^{q} \alpha_p \Bigg[ \dfrac{1}{\sqrt{{\pi}^n}}  \sum_{\cup_{a_{i,j}}} \Bigl\{ \prod_{i=1}^{n} \big( D_i \hat{x}_{i}^{a_{i,1}} ( \prod_{j=1}^{n}S_{ij}^{a_{i,j+1}} ) \\
   & \ \ (2c_i)^{(\dfrac{1}{2}\sum_{j=1}^n a_{j,i+1} )} 
			 \Gamma(\dfrac{\sum_{j=1}^n a_{j,i+1}+1}{2})     \big)\Bigr\}\Bigg],
        \end{split}
    \end{equation}
 where $\cup_{a_{i,j}}$ $(i=1,2,\ldots n),\ (j=1,2\ldots n+1)$ are all possible combinations of $a_{i,j}$ which satisfies the two conditions subsequently given as:
(i) when $i\neq g$, $\sum_{j=1}^{n+1}a_{i,j} = m_{i,p}$ and $D_i=\dfrac{m_{i,p}!}{a_{i,1}!a_{i,2}\,!\ldots a_{i,n+1}!}$,  for $i= g$,   $\sum_{j=1}^{n+1}a_{i,j} = m_{i,p}+1$, and $D_i= \dfrac{(m_{i,p}+1)!}{a_{i,1}!a_{i,2}\,!\ldots a_{i,n+1}!}$, 
  (ii) $\sum_{j=1}^n a_{j,i+1} =\text{even}$.  $c_i$ is the eigenvalue of the covariance matrix $P_{xx}$, and $S$ is the orthogonal matrix satisfying the condition $S^{-1}P_{xx}S=\diag(d_1,d_2\ldots d_n)$, $S_{i,j}$ is the $(i,j)^{th}$ element of matrix $S$.
\begin{proof}
    The integral $\mathcal{I}$ from \eqref{eqnfg} can be further expressed by substituting $f_{g'}(x_k)= \sum_{p=1}^{q} \alpha_p \prod_{i=1}^n x_{i}^{m_{i,p}}$, as
    \begin{equation*}
              \mathcal{I}  =  \int\limits_{-\infty}^{\infty} x_{g,k}\sum_{p=1}^{q} \alpha_p \prod_{i=1}^n x_{i,k}^{m_{i,p}}\mathcal{N}(x_k;\hat{x}_{k|k},P_{xx,k|k})dx_k,\end{equation*}\begin{equation*}
        =  \int\limits_{-\infty}^{\infty} \sum_{p=1}^{q} \alpha_p \prod_{i=1}^n x_{g,k} x_{i,k}^{m_{i,p}}\mathcal{N}(x_k;\hat{x}_{k|k},P_{xx,k|k})dx_k.
    \end{equation*}
Using Theorem 1, the above integral becomes \eqref{eqngi2}.
\end{proof}
    
\end{theorem}
\vspace{-0.25 cm}
The resulting GIRTSS is presented in Algorithm \ref{Algo1_GIRTSS}.
\vspace{-0.25 cm}
\begin{algorithm}
		\caption{Gaussian integral-based RTS smoother}\label{Algo1_GIRTSS}
		\begin{algorithmic}[1]
  \Function{$[\hat{x}_{k\mid T}^s, \, P_{k\mid T}^s] = \text{GIRTSS}$}{$\hat{x}_{0\mid 0},\, P_{xx,0\mid 0}$}. 
  \For{$k = 1, \ldots, T$}
  \State 	$\begin{bmatrix}S_{k-1|k-1}, & c_{k-1|k-1}\end{bmatrix}= eig(P_{xx,k-1|k-1})$.
  \State $\hat{x}_{k|k-1} = \int f(x_{k-1})p(x_{k-1}|y_{k-1}) dx_{k-1}$ using \eqref{eqngi}.
  \State $P_{xx,k|k-1} = \int f(x_{k-1})f(x_{k-1})^\top p(x_{k-1}|y_{k-1}) dx_{k-1} $
  \Statex \hspace{2.3 cm} $- \hat{x}_{k|k-1}\hat{x}_{k|k-1}^\top+Q_{k-1}$ using \eqref{eqngi2}.
 \State $ \hat{y}_{k|k-1} =  \int h(x_k) p(x_k|y_{k-1}) dx_k $ using \eqref{eqngi}.
 \State $P_{yy,k|k-1} = \int h(x_k)h(x_k)^\top p(x_k|y_{k-1}) dx_k $
 \Statex \hspace{2.3 cm}$-  \hat{y}_{k|k-1}\hat{y}_{k|k-1}^\top+R_k$ using \eqref{eqngi2}.
\State $P_{xy,k|k-1} = \int x_kh(x_k)^\top p(x_k|y_{k-1})dx_k$
\Statex \hspace{2.3 cm}$-\hat{x}_{k|k-1}y_{k|k-1}^\top$ using \eqref{eqngi2}.
\State $K_k = P_{xy,k|k-1}(P_{yy,k|k-1})^{-1}$.
  \State $\hat{x}_{k|k} =  \hat{x}_{k|k-1}-K_k(y_k-\hat{y}_{k|k-1})$.
 \State $P_{xx,k|k} =  P_{xx,k|k-1}-K_kP_{yy,k|k-1}K_k^\top$.
\EndFor
\State $\hat{x}_{T\mid T}^s = \hat{x}_{T\mid T} $ and $P_{xx,T\mid T}^s = P_{xx,T\mid T}$. 
  \For{$k = T-1, \ldots, 1$}
\State 	$\begin{bmatrix}S_{k-1|k-1}, & c_{k-1|k-1}\end{bmatrix}= eig(P_{xx,k|k})$.
  \State Calculate $P_{x_k,x_{k+1|k}}$ from \eqref{eqnrtscross} using \eqref{eqngi2}.
  \State Compute $\hat{x}_{k|T}^s$ and $P_{xx,k|T}^s$ using \eqref{eqnrtsx} and \eqref{eqnrtsp}, \Statex \hspace{0.4 cm} respectively. 
\EndFor
\EndFunction
  \end{algorithmic}
  \end{algorithm} 
\begin{figure*}
	\centering
	\begin{subfigure}[b]{0.316\textwidth}
		\centering
		\includegraphics[width=\textwidth,height = 3.8 cm]{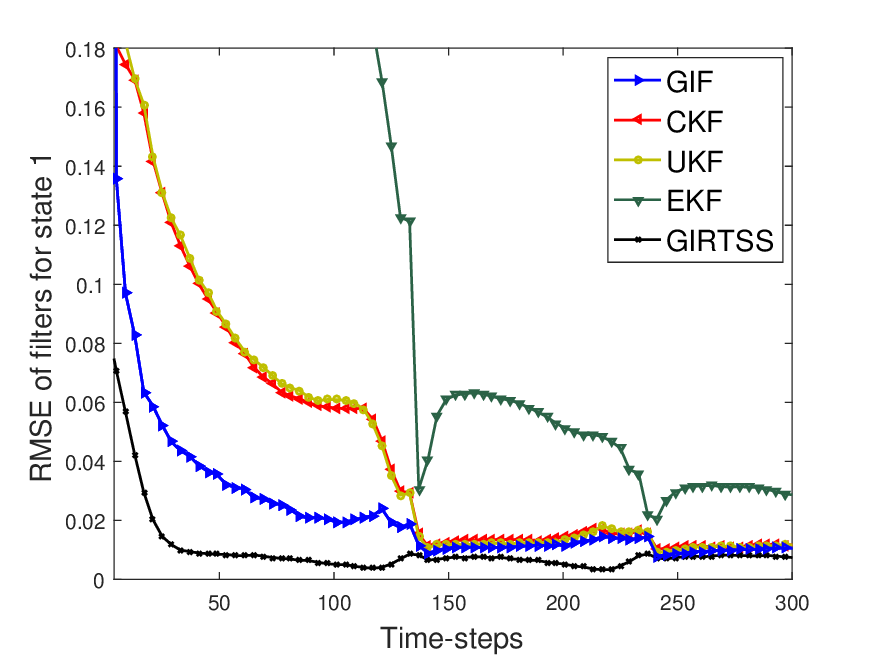}
	\end{subfigure}
 \begin{subfigure}[b]{0.316\textwidth}
		\centering
		\includegraphics[width=\textwidth,height = 3.8 cm]{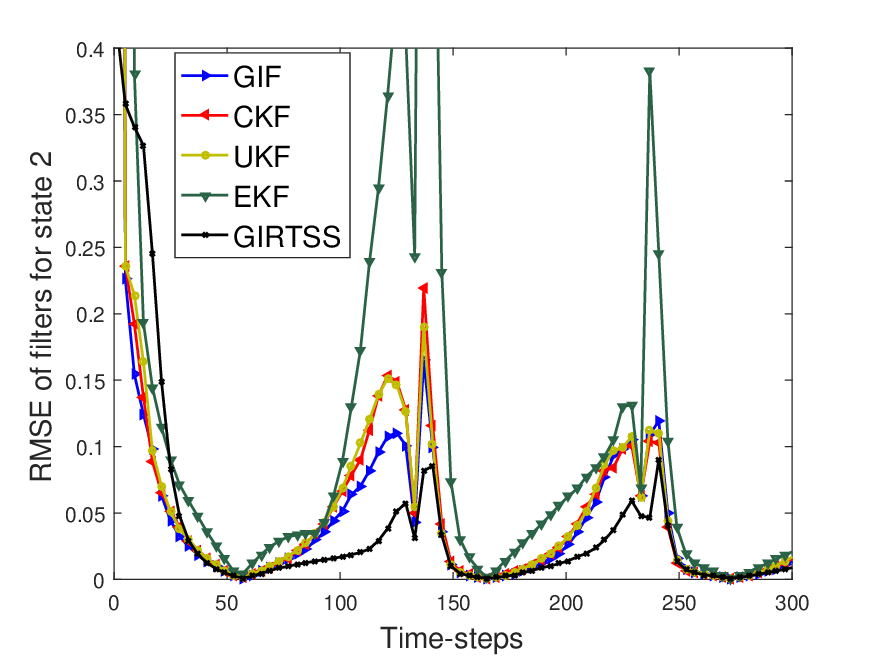}
	\end{subfigure}
	\begin{subfigure}[b]{0.316\textwidth}
		\centering
		\includegraphics[width=\textwidth,height = 3.8 cm]{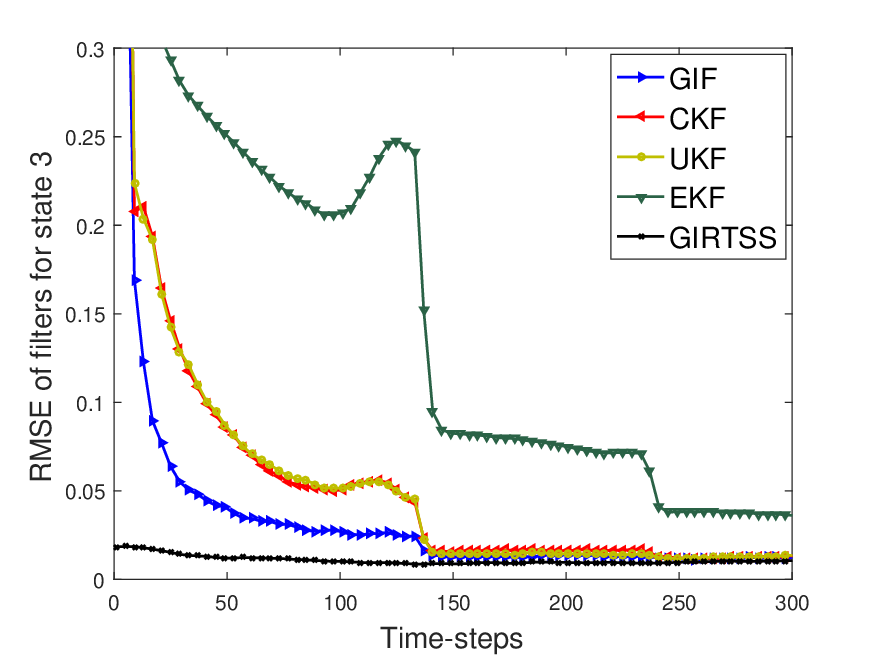}
	\end{subfigure}
	\caption{RMSE of filters for state 1, state 2 and state 3.  }
	 \label{figvdprmse1}
\end{figure*}
\begin{figure*}
	\centering
	\begin{subfigure}[b]{0.316\textwidth}
		\centering
		\includegraphics[width=\textwidth,height = 3.8 cm]{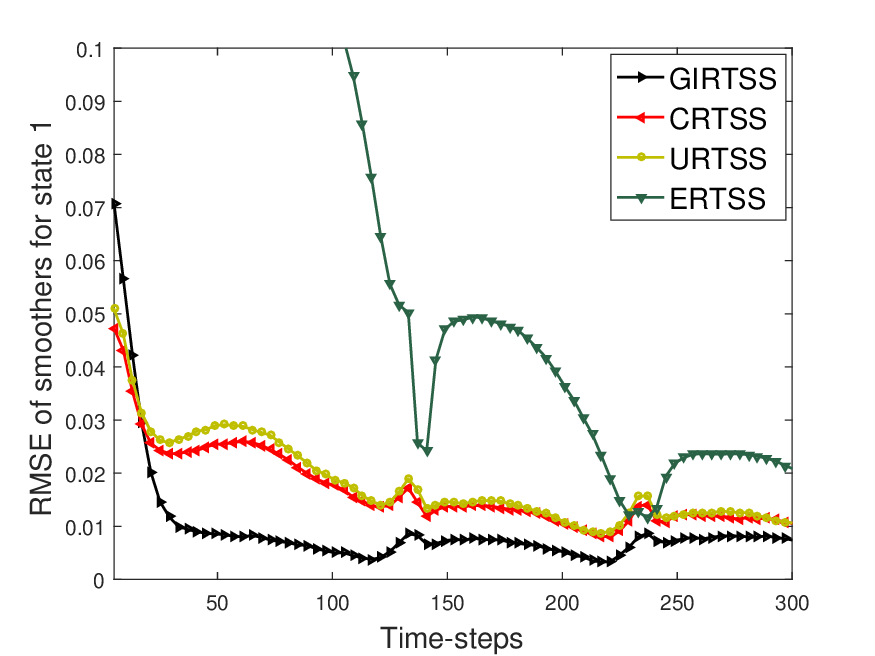}
	\end{subfigure}
 \begin{subfigure}[b]{0.316\textwidth}
		\centering
		\includegraphics[width=\textwidth,height = 3.8 cm]{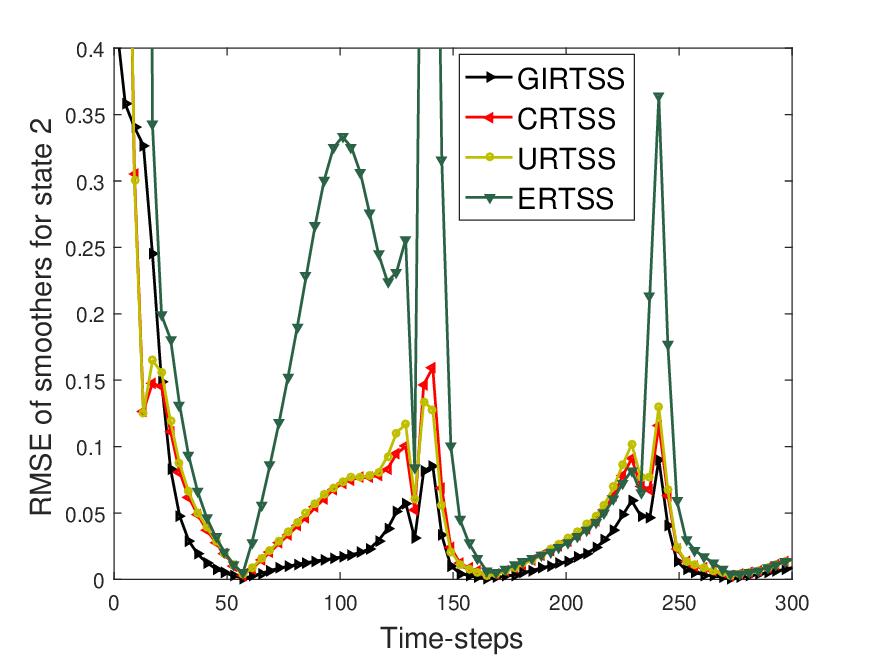}
	\end{subfigure}
	\begin{subfigure}[b]{0.316\textwidth}
		\centering
		\includegraphics[width=\textwidth,height = 3.8 cm]{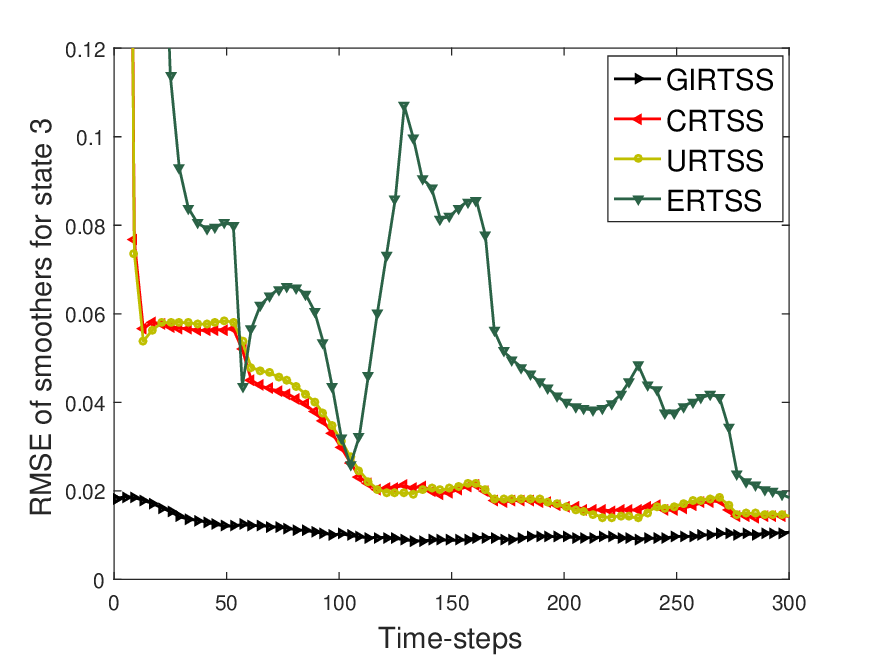}
	\end{subfigure}
	\caption{RMSE of smoothers for state 1, state 2 and state 3.  }
	 \label{figvdprmse2}
\end{figure*}
\vspace{-0.4 cm}
\section{Simulation and Results}
In this section, we implement the proposed GIRTSS to obtain smoothed states for the Van der Pol oscillatory systems. 
VDP oscillator is described by a second-order differential equation given below
\begin{equation} \label{eqnvdpsys}
\ddot{\mathcal{X}}-\zeta(1-\mathcal{X}^2)\dot{\mathcal{X}}+\mathcal{X}= \mathcal{F}(t), 
\end{equation}
where $\mathcal{X}$ and $\zeta$ are the displacement and friction damping coefficient, respectively. $\mathcal{F}(t) = A\cos(\lambda t)$ is the forcing input, where $A$ and $\lambda$ are amplitude and frequency of oscillation, respectively. The states variables are \cite{onat2019novel} $x(t)=\begin{bmatrix}x_1 & x_2 & x_3 \end{bmatrix}^\top=\begin{bmatrix} \mathcal{X} & \dot{\mathcal{X}} & \zeta \end{bmatrix}^\top$. 
VDP is represented in the form of state space model from \eqref{eqnprocess} and \eqref{eqnmeasurement}. The discrete-time state evaluation function is $f(x_k) = \big[ x_{1,k} + x_{2,k}\delta \ \ x_{2,k}+(x_{3,k}(1-x_{1,k}^2)x_{2,k} -x_{1,k}+A\cos(\lambda k\delta) )\delta  \ \ x_{3,k} \big]^\top$, and the measurement function is $h(x_k) = \begin{bmatrix}x_{1,k}&  x_{2,k}\end{bmatrix}^\top$. The process and measurement noise covariance matrix are $Q_k = 10^{-3}\times I_3$ and $R_k =10^{-1} \times I_2$, respectively. $\delta$ is the sampling interval.

\textit{Implementation and performance analysis:}
The performance of the proposed GIRTSS is compared with popular existing smoothers, including ERTSS, URTSS, and CRTSS, using performance metrics such as root mean square error (RMSE), average RMSE, and relative execution time (RET). These algorithms are applied to the state smoothing of the Van der Pol (VDP) oscillator. GIRTSS is implemented using Algorithm \ref{Algo1_GIRTSS}, URTSS is implemented with parameter $\kappa = -1$, CRTSS with a 3rd-degree spherical-radial rule, and ERTSS with a first-order Taylor series approximation of the non-linear function.
Simulated data for true states and measurements is generated using the initial true state $x_{1} = \begin{bmatrix} 2.75 & 0 & 2 \end{bmatrix}^\top$. The system parameters for modeling the VDP oscillator dynamics are $A = 100$ and $\lambda = (1.85\pi)/2$ \cite{yamalakonda2023oscillatory}. The initial state estimate is $\hat{x}_{0|0} = \begin{bmatrix} 0 & -3 & 1\end{bmatrix}^\top$, with an error covariance of $P_{xx,0|0} = \text{diag}(\begin{bmatrix} 10 & 10 & 0.5 \end{bmatrix})$. The sampling interval is $\delta = 0.01$, and the simulation is conducted over 300 time steps with 1000 Monte Carlo runs.

The RMSE for the forward filters such as Gaussian integral filter (GIF), cubature Kalman filter (CKF), unscented Kalman filter (UKF), and extended Kalman filter (EKF), along with the RMSE for GIRTSS, for 3 system states of VDP oscillator is plotted in Fig. \ref{figvdprmse1}. 
 The GIF achieves the lowest RMSE among the filters. Notably, the RMSE for GIF and GIRTSS are the same at the final time step, but the RMSE for GIRTSS decreases as we move backwards in time, illustrating the smoothers enhanced performance.
 
The RMSE plots for smoothers are shown in Fig. \ref{figvdprmse2}.
It is observed that the proposed GIRTSS smoother achieves the lowest RMSE, followed by URTSS, CRTSS, and ERTSS.
This superior performance is due to GIRTSS evaluating the exact Gaussian integration for polynomial functions over a Gaussian PDF instead of approximating it.
 Notably, the RMSE values for the smoothers at the final time step match those of their corresponding filters. However, as time moves backwards, the smoothers RMSE values converge to lower levels compared to the filters.
 The average RMSEs for the three states (S1, S2, S3) are provided in Table \ref{tablevdp}, confirming that the GI-based approach is the best-performing, followed by CRTSS, URTSS, and ERTSS. The RET for filters and smoothers relative to EKF is shown in Table \ref{tablevdp}. The GI-based algorithm has a higher RET due to the computationally expensive nature of the exact GI.
    \begin{table}[H]
    \caption{Average RMSE and relative execution time (RET) for different filters and smoothers.}
    \begin{center}
        \setlength{\tabcolsep}{0.148cm}  
        \begin{tabular}{|c@{\hskip 0.18cm}c@{\hskip 0.25cm}c@{\hskip 0.25cm}c@{\hskip 0.17cm}c|c@{\hskip 0.1cm}c@{\hskip 0.25cm}c@{\hskip 0.25cm}c@{\hskip 0.17cm}c|}
            \hline
            Filter & S1 & S2 & S3 & RET  & Smoother & S1 & S2 & S3 & RET\\
            \hline
            GIF & 0.026 & 0.049 & 0.032 & 5.03 & GIRTSS & 0.009 & 0.034 & 0.009 & 6.33 \\
            CKF & 0.040 & 0.055 & 0.047 & 1.63 & CRTSS & 0.014 & 0.042 & 0.040 & 1.96\\
            UKF & 0.044 & 0.057 & 0.051 & 1.63 & URTSS & 0.014 & 0.044 & 0.044 & 1.962 \\
            EKF & 0.177 & 0.357 & 0.143 & 1 & ERTSS & 0.136 & 0.264 & 0.089 & 1.38 \\
            \hline
        \end{tabular}
        \label{tablevdp}
    \end{center}
\end{table}
\vspace{-0.7cm}
\section{Conclusion}
In this work, we proposed a GI-based RTS smoothing algorithm that computes the mean and covariance of system states using the odd-even properties of Gaussian integrals. The algorithm performs forward and backward filtering, with the backward process refining the forward estimates. Unlike traditional smoothing algorithms that rely on approximations for intractable integrals, the GI-based smoother provides an exact solution for the integral of polynomial functions over a Gaussian PDF, leading to more accurate results. This is demonstrated through the implementation of the proposed smoother to the Van der Pol oscillator, where the proposed algorithm shows improved accuracy compared to state-of-the-art smoothers like CRTSS, URTSS, and ERTSS. 
\newpage


\end{document}